\documentclass[sigplan]{acmart}

\acmConference[]{Preprint}{Under Review}{January 10, 2020}
\acmYear{2020}
\acmISBN{} 
\acmDOI{} 
\startPage{1}

\setcopyright{none}

\usepackage{booktabs}   
\usepackage{subcaption} 

\usepackage{bm}
\usepackage{enumitem}
\usepackage{tikz}
\usepackage{tikz-3dplot} 
\usetikzlibrary{shapes}

\usepackage{kbordermatrix}

\usepackage{myabbrevmacros}
\usepackage{myutilmacros}

\usepackage{algorithm}
\usepackage[noend]{algpseudocode}

\newcommand{\two}{\ensuremath{\mathbf{2}}}

\newcommand{\setfunc}{\ensuremath{\mathbf{Set}}}

\newcommand{\mvconst}{\ensuremath{\delta}}
\newcommand{\consts}{\ensuremath{\Delta}}
\newcommand{\mvaconst}{\ensuremath{\gamma}}
\newcommand{\aconsts}{\ensuremath{\Gamma}}

\newcommand{\dnf}{\ensuremath{\operatorname{dnf}}}

\newcommand{\belief}{\ensuremath{\mathbb{B}}}

\newcommand\vecsp{\mathbb{V}}
\newcommand\hilbsp{\mathbb{H}}

\newcommand{\inner}[2]{\ensuremath{\langle#1, #2\rangle}}

\begin{document}

\title[Elementary Logic in Linear Space]{Elementary Logic in Linear Space}         


\author{Daniel Huang}
\orcid{nnnn-nnnn-nnnn-nnnn}             
\affiliation{
  \institution{University of California, Berkeley, EECS}            
}
\email{dehuang@berkeley.edu}          


\begin{abstract}
First-order logic is typically presented as the study of deduction in a setting with elementary quantification. In this paper, we take another vantage point and conceptualize first-order logic as a linear \emph{space} that encodes ``plausibility". Whereas a deductive perspective emphasizes \emph{how} (\ie, process), a space perspective emphasizes \emph{where} (\ie, location). We explore several consequences that a shift in perspective to ``signals in space" has for first-order logic, including (1) a notion of proof based on orthogonal decomposition, (2) a method for assigning probabilities to sentences that reflects logical uncertainty, and (3) a ``models as boundary" principle that relates the models of a theory to its ``size".
\end{abstract}

\ccsdesc[500]{Software and its engineering~General programming languages}
\ccsdesc[300]{Social and professional topics~History of programming languages}

\keywords{first-order logic, linear space, logical uncertainty, distributive normal form}  

\maketitle

\section{Introduction}
\label{sec:intro}

The inspiration for this paper has its origins in P{\'o}lya's writings on \emph{plausible reasoning}~\cite{polya1990mathematics1,polya1990mathematics2,polya2004solve}.
\begin{quote}
    Finished mathematics presented in a finished form appears as purely demonstrative, consisting of proofs only. Yet mathematics in the making resembles any other human knowledge in the making. You have to guess a mathematical theorem before you can prove it; you have to guess the idea of the proof before you carry through the details. You have to combine observations and follow analogies; you have to try and try again.~\citep[][pg. vi]{polya1990mathematics1}
\end{quote}
This leads us to ask: \emph{Where is the information content of logical experimentation---the guesses, observations, and analogies---located?} In this paper, we give our attempt at answering this question in the restricted setting where we use a first-order language to express statements. 

The standard presentation of first-order logic is as the study of deduction in a setting with elementary (\ie, over elements of sets) quantification. Thus our first task is to identify a space that holds the information content of first-order sentences. Towards this end, we build upon ideas found in~\cite{huang2019oltp} which introduces a probability distribution on first-order sentences and identifies a corresponding space. We organize our exploration as a series of questions.
\begin{itemize}
    \item \emph{Is there a space that encodes the information content of first-order sentences with fixed quantifier rank (Section~\ref{sec:statesp})?}
    We use the observation that every first-order sentence can be written in \emph{distributive normal form}~\cite{hintikka1965distributive} and identify a \emph{vector space} where sentences of fixed quantifier rank are represented as vectors. 
    \item \emph{Is there a space for the case of sentences with arbitrary quantifier rank (Section~\ref{sec:statesp2})?} 
    We will see that an infinite dimensional analogue of a vector space---a \emph{Hilbert space}---is an appropriate generalization to the setting with arbitrary quantifier rank. A first-order sentence can be identified with a signal in this space.
    \item \emph{How does deduction affect the ``plausibility" of a logical system (Section~\ref{sec:plaus})?}
    We can conceptualize a logical system as carrying the ``plausibility" of first-order sentences. Similar to how a physical system carries an invariant amount of energy in different forms, a logical system carries an invariant amount of ``plausibility" in different places. Under this formulation, we will see that deduction is an \emph{entropic} operation.
\end{itemize}  
At this point, we will have a view of first-order logic as ``signals in space" and so can turn our attention towards the \emph{geometry} of sentences in space.
\begin{itemize}
    \item \emph{Where is a first-order theory located (Section~\ref{sec:theory})?}
    A theory manifests itself as a subspace of a certain dimension. It turns out that a complete theory has dimension one whereas an undecidable theory's dimension is not computable.
    \item \emph{What does orthogonal decomposition reveal about proofs (Section~\ref{sec:proofs})?}
    We show how a proof can be constructed by decomposing the ``plausibility" of a first-order sentence into its orthogonal components. Counter-examples (\ie, theorem proving) can be used to approximate a theory from below and \emph{examples} (\ie, model checking) can be used to approximate a theory from above. An application to assigning probabilities to first-order sentences that reflect \emph{logical uncertainty} and a discussion of conjecturing will also be given.
    \item \emph{How much space does a first-order theory occupy (Section~\ref{sec:size})?}
    We explore two aspects of a ``models as boundary" principle that highlight the difficulty of approximating a theory. In particular, we will encounter a familiar (edge) isoperimetric principle on the Boolean hypercube: a theory with small ``volume" (\ie, variance) can have unusually large ``perimeter" (\ie, many models).
\end{itemize}
We emphasize that our aim in this paper is to offer a complementary perspective on first-order logic as ``signals in space". At the end of the paper, we will discuss some directions for future work.

\paragraph{Preliminaries}
We assume familiarity with the syntax and semantics of first-order logic, vector spaces, and basic probability theory.

As notation, we write a formula $\mvform$ with free variables $y_1, \dots, y_k$ as $\mvform[\bar{y}_k]$ or simply $\mvform[\bar{y}]$. When the quantifier rank of a sentence is important, we indicate it as $\mvform^{(r)}$. Throughout the paper, we work with a first-order language $\cL$ with equality and a finite number of predicates and without constants or function symbols.

Let $\two \eqdef \set{0, 1} \cong \set{\false, \true}$ denote the two element set where $\true \eqdef (\forall x) \, x = x \lor \lnot ((\forall x) \, x = x)$ and $\false \eqdef \lnot \true$. Let $\setfunc(X) \eqdef \two^X$ denote the powerset of $X$.

\section{Space: Fixed Quantifier Rank}
\label{sec:statesp}

\emph{Is there a space that encodes the information content of first-order sentences with fixed quantifier rank?}
We begin our exploration of first-order logic as a space starting with the simpler case of fixed quantifier rank. We use the observation that every first-order sentence can be written in \emph{distributive normal form}~\cite{hintikka1965distributive} and identify a \emph{vector space} where sentences of fixed quantifier rank are represented as vectors (Section~\ref{subsec:statesp:vecsp}). We then examine basic properties of the vector space (Section~\ref{subsec:statesp:structure}). Throughout this section, we restrict attention to sentences of quantifier rank $r$.

\subsection{Vector Space}
\label{subsec:statesp:vecsp}

We define a vector space for rank $r$ sentences by identifying a basis of rank $r$ \emph{constituents} for the vector space. Before we review the definition of constituents, we highlight two properties that make them a natural candidate for a basis.
\begin{proposition}[\cite{hintikka1965distributive}]
\begin{description}
    \item[Distributive normal form] Every formula $\mvform^{(r)}[\bar{y}]$ can be written in \emph{distributive normal form}, \ie, as a disjunction
    \[
    \mvform^{(r)}[\bar{y}] \equiv \lOr_{\mvconst^{(r)}[\bar{y}] \in \dnf(\mvform^{(r)}[\bar{y}])} \mvconst^{(r)}[\bar{y}]
    \]
    where $\dnf: \cL[\bar{y}] \rightarrow \setfunc(\consts^{(r)}[\bar{y}])$ is a function that maps a formula (an element of $\cL[\bar{y}]$) to a subset of constituents (an element of $\setfunc(\consts^{(r)}[\bar{y}])$).
    \item[Mutual exclusion] Any two distinct constituents $\mvconst^{(r)}_i[\bar{y}]$ and $\mvconst^{(r)}_j[\bar{y}]$ are mutually exclusive, \ie, $\vDash \mvconst^{(r)}_i[\bar{y}] \rightarrow \lnot \mvconst^{(r)}_j[\bar{y}]$.
\end{description}
\label{prop:statesp:vecsp:constituent}
\end{proposition}
\noindent The first item captures the idea that every vector (formula) can be expressed as a linear combination of basis vectors (constituents) where logical or is interpreted as vector addition. The second item hints suggests an inner product for the vector space: two formulas that denote logically distinct possibilities will be orthogonal. Thus constituents actually give an orthogonal basis. We introduce constituents concretely now following~\cite{hintikka1965distributive}.

The intuitive idea behind the definition of a constituent is that we would like to describe possible kinds of ``logical worlds" by enumerating descriptions of how $r$ individuals in the domain are related to one another uniformly in $r$.\footnote{Notably, this method of describing possibilities yields a finite number of finite descriptions for each quantifier rank $r$. In contrast, attempting to describe possibilities by enumerating the individuals in the domain may result in an infinite description when the domain is infinite.} In order to define constituents, we require (1) several auxiliary definitions and (2) the notion of an \emph{attributive constituent}. We introduce these in turn.

\paragraph{Auxiliary definitions}
Let $\cA[\bar{y}]$ denote the set of all atomic formula (\ie, a predicate applied to a tuple of variables) involving the free variables $\bar{y}$. Let $\cB[\bar{y}]$ denote the subset of $\cA[\bar{y}]$ that mentions the last element of $\bar{y}$ at least once. Let
\[
\mathbf{S}(\set{\mvform_1, \dots, \mvform_k}) \eqdef \set{ \lAnd_{i \in \set{1, \dots, k}} (\pm)^{b_i} \mvform_i \ST b_1 \in \two, \dots, b_k \in \two}
\]
be the set where an element is a conjunction of every $\mvform_i$ or its negation where $(\pm)^0 \mvform \eqdef \lnot \mvform$, and $(\pm)^1 \mvform \eqdef \mvform$.

\paragraph{Attributive constituents}
The set of attributive constituents $\aconsts^{(r)}[\bar{y}]$ with free variables $\bar{y}$ is defined by induction on quantifier rank $r$.

In the base case, 
\[
\aconsts^{(0)}[\bar{y}] \eqdef \mathbf{S}(\cB[\bar{y}]) \,.
\]
Intuitively, a rank $0$ attributive constituent with free variables $\bar{y} = y_1, \dots, y_k$ describes how the individual $y_k$ relates to the other individuals $y_1, \dots, y_{k-1}$ via all the atomic formula in the language.

In the inductive case, 
\begin{multline*}
    \aconsts^{(r)}[\bar{y}] \eqdef \Big\{ \mvaconst^{(0)}[\bar{y}] \land \lAnd_{\mvaconst[\bar{y};z] \in \aconsts^{(r-1)}[\bar{y};z]} (\pm)^{s(\mvaconst[\bar{y};z]) } (\exists^E z) \mvaconst[\bar{y};z] \\
     \ST \mvaconst^{(0)} \in \aconsts^{(0)}[\bar{y}], s: \aconsts^{(r-1)}[\bar{y};z] \rightarrow \two \Big\} \,.
\end{multline*}
The notation $(\exists^E z)$ indicates an exclusive interpretation of quantification so that $z$ is distinct from all other variables in scope. For example, $(\exists^E z) \mvform[\bar{y}; z] \eqdef (\exists z) z \neq y_1 \land \dots \land z \neq y_k \land \mvform[\bar{y}; z]$ where $x \neq y$ is shorthand for $\lnot (x = y)$. Intuitively, a rank $r$ constituent with free variables $\bar{y}$ describes (1) how $y_k$ is related to each $y_1, \dots, y_{k-1}$ where $\bar{y} = y_1, \dots, y_k$ and (2) for every possible smaller description (\ie, a rank $r-1$ attributive constituents), whether or not an additional individual $z$ exists that is related to the other individuals $\bar{y}$ recursively via that description.

\paragraph{Constituents}
The set of rank $r$ constituents $\consts^{(r)}[\bar{y}]$ with free variables $\bar{y}$ is defined using attributive constituents as
\begin{multline*}
\consts^{(r)}[\bar{y}; z] \eqdef \{A[\bar{y}] \land \mvaconst^{(r)}[\bar{y}; z] \\
\ST A[\bar{y}] \in \cA[\bar{y}], \mvaconst^{(r)}[\bar{y}; z] \in \aconsts^{(r)}[\bar{y}; z] \} \,.    
\end{multline*}
Thus a rank $r$ constituent with free variables $\bar{y}; z$ completes the description of a rank $r$ attributive constituent by additionally describing how each term in $\bar{y}$ is related to one another. Thus we can think of a constituent as describing a possible kind of ``logical world".

Some constituents will be unsatisfiable, and hence, describe possibilities that are impossible. The satisfiable constituents, on the other hand, describe logically possible worlds. 
\begin{proposition}[\cite{hintikka1965distributive}]
A formula $\mvform^{(r)}[\bar{y}]$ is logically valid iff its distributive normal form contains all satisfiable rank $r$ constituents.
\end{proposition}
\noindent Of course, it is not decidable whether or not a constituent is satisfiable because validity of first-order statements is not decidable. 
When there are no free variables, attributive constituents and constituents are identical. We abbreviate $\consts^{(r)}[]$ as $\consts^{(r)}$. Throughout the rest of the paper, we will largely focus on the case with no free variables, \ie, sentences, unless indicated otherwise.\footnote{A constituent $\mvconst^{(r)}[\bar{y}_k]$ is equivalent to $(\exists y_1) \dots (\exists y_k) \mvconst^{(r)}[\bar{y}_k]$ which has a distributive normal form of rank $r+k$.}

\paragraph{Constituents as a basis}
As a reminder, a (real-valued) vector space $(V, +, \cdot)$ is a set $V$ along with vector addition $+: V \times V \rightarrow V$ and scalar multiplication $\cdot: \R \times V \rightarrow V$ satisfying the vector space axioms. We define a vector space for sentences of quantifier rank $r$ by defining the set and the two operations.
\begin{definition}
Let $V \eqdef \set{ a_1 \hat{\mvconst}^{(r)}_1 + \dots + a_n \hat{\mvconst}^{(r)}_n \ST \set{\mvconst^{(r)}_1, \dots \mvconst^{(r)}_n} \in \setfunc(\consts^{(r)}), a_1 \in \R, \dots, a_n \in \R }$. Define $+$ as component-wise addition and $\cdot$ as component-wise multiplication.
\end{definition}
\begin{proposition}
$\vecsp^{(r)} \eqdef (V, +, \cdot)$ is a vector space with zero $\vec{0}$ given by the all $0$ vector and additive inverse given by the negation of the vector components.
\end{proposition}
\begin{proof}
Routine verification of the vector space axioms.
\end{proof}
\begin{proposition}
Constituents form a basis for $\vecsp^{(r)}$.
\end{proposition}
\begin{proof}
By construction.
\end{proof}
\noindent Because a vector space can be described by many different bases, we should think of constituents as the \emph{standard basis}, similar to the unit vectors of $\R^d$.

The vector space spans all rank $r$ first-order sentences.
\begin{proposition}
Every (rank $r$) first-order sentence $\mvform^{(r)}$ can be written as
\[
\bm{\mvform}^{(r)} = a_1 \hat{\mvconst}^{(r)}_1 + \dots a_n \hat{\mvconst}^{(r)}_n
\]
where $\dnf(\mvform^{(r)}) = \set{\mvconst^{(r)}_1, \dots, \mvconst^{(r)}_n}$ and any non-zero $a_1, \dots, a_n$. Unless stated otherwise, we assume $a_1 = \dots = a_n = |\dnf(\mvform^{(r)})|$.
\end{proposition}
\begin{proof}
Every first-order sentence of rank $r$ can be written as a disjunction of rank $r$ constituents (Proposition~\ref{prop:statesp:vecsp:constituent}).
\end{proof}

As a reminder, every sentence of rank $s \leq r$ is equivalent to some sentence of rank $r$. Thus this vector space expresses all first-order sentences up to rank $r$.

\subsection{Structure}
\label{subsec:statesp:structure}

We review some of the structure of the vector space in this section.

\paragraph{Dimensions}
The number of attributive constituents of rank $r$ can be computed recursively as
\[
|\aconsts^{(r)}[\bar{y}]| \eqdef |\aconsts^{(0)}[\bar{y}]| 2^{|\aconsts^{(r-1)}[\bar{y}; z]|}
\]
where $|\aconsts^{(0)}[\bar{y}]|$, is combinatorial in the number of base predicates in $\bar{y}$. The number of constituents of rank $r$ is then
\[
|\consts^{(r)}[\bar{y}]| \eqdef |\consts^{(0)}[\bar{y}]| 2^{|\aconsts^{(r-1)}[\bar{y}; z]|}
\]
where $|\consts^{(0)}[\bar{y}]|$ is combinatorial in the number of base predicates in $\bar{y}$. Consequently, the first-order vector space has a super-exponential number of dimensions as a function of quantifier rank $r$. Thus the vector space has astronomic dimension. Later in Section~\ref{sec:size}, we will see how models of a theory can be used to approximately span the theory.

\paragraph{Orthogonality}
We can define an inner product between two vectors in the standard way:
\[
\inner{v}{w} \eqdef \sum_{\mvconst^{(r)} \in \consts^{(r)}} v_{\mvconst^{(r)}} w_{\mvconst^{(r)}} \,.
\]
It is a routine exercise to check that this is an inner product. As usual, two vector are orthogonal if their inner product is $0$.
\begin{proposition}
Two vectors are orthogonal in the space if they are logically mutually exclusive.
\end{proposition}
\begin{proof}
By Proposition~\ref{prop:statesp:vecsp:constituent}.
\end{proof}

The converse does not hold because there can be multiple inconsistent constituents at some quantifier rank $r$. In other words, there are multiple ways to express falsehood in this standard basis so we may have two unsatisfiable sentences that have an inner product of $0$. Note that if we remove all inconsistent constituents of rank $r$ from the basis, then we end up with a basis of $r$-isomorphism types (\eg, see~\cite{libkin2004fmt}). In this case, falsehood is represented by $\vec{0}$.

\paragraph{Gram-Schmidt}
Let $\cE$ be an enumeration $\cL^{(r)}$ that spans it, \ie, every first-order sentence of rank $r$ can be written as a disjunction of elements from $\cE$. Given an enumeration of $\cL^{(r)}$, we can construct an enumeration that spans it as $\lnot \mvform^{(r)}_1, \mvform^{(r)}_1,  \lnot \mvform^{(r)}_1 \land \lnot \mvform^{(r)}_2, \lnot \mvform^{(r)}_1 \land \mvform^{(r)}_2, \mvform^{(r)}_1 \land \lnot \mvform^{(r)}_2, \mvform^{(r)}_1 \land \mvform^{(r)}_2, \dots$.

Recall that the Gram-Schmidt process is a method for constructing an orthonormal basis from a set of vectors that spans the space. Applying the Gram-Schmidt process results in the following basis: $\hat{e_1} \eqdef \mvform^{(r)}_1, \hat{e_2} \eqdef \lnot \mvform^{(r)}_1 \land \mvform^{(r)}_2, \hat{e_3} \eqdef \lnot \mvform^{(r)}_1 \land \lnot \mvform^{(r)}_2 \land \mvform^{(r)}_3, \dots$.

\newcommand{\mvval}{\nu}
\newcommand{\mvmeas}{\beta}

\section{Space: Arbitrary Quantifier Rank}
\label{sec:statesp2}

\emph{Is there a space for the case of sentences with arbitrary quantifier rank?} 
The case of sentences with fixed quantifier rank suffices to illustrate the concepts that we hope to explore in the rest of the paper. Nevertheless, for the sake of completeness, we generalize the space perspective for sentences with fixed quantifier rank to the case where sentences have arbitrary quantifier rank. Towards this end, we identify a \emph{Hilbert space} where first-order sentences are represented as \emph{signals} in this space. 

The construction of the Hilbert space is inspired by the one given in~\cite{huang2019oltp} (Section~\ref{subsec:statesp2:ms} and Section~\ref{subsec:statesp2:hilbert}). The main difference is that we take the space view as primary whereas~\cite{huang2019oltp} defines a probability distribution on first-order sentences and derives the corresponding space.

\subsection{Limits of Logical Descriptions}
\label{subsec:statesp2:ms}

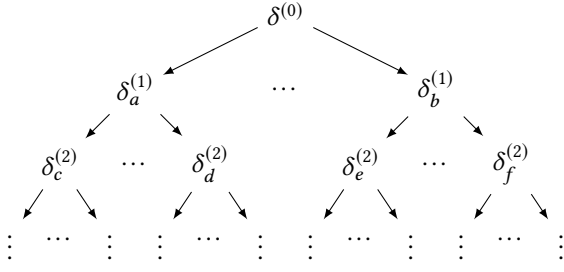
\begin{figure}[t]
    \centering
    \begin{tikzpicture}[edge from parent/.style={draw,-latex},level/.style={level distance=1cm, sibling distance=40mm/#1}]
\node [] (z){$\mvconst^{(0)}$}
child {
  node (a) {$\mvconst^{(1)}_{a}$}
  child {
    node (b) {$\mvconst^{(2)}_{c}$}
      child { node (c) {$\vdots$} } 
      child { node (d) {$\vdots$} }
  }
  child {
    node (g) {$\mvconst^{(2)}_{d}$}
      child { node (e) {$\vdots$} }
      child { node (f) {$\vdots$} }
  }
}
child {
  node (j) {$\mvconst^{(1)}_{b}$}
  child {
    node (k) {$\mvconst^{(2)}_{e}$}
      child { node (m) {$\vdots$} } 
      child { node (n) {$\vdots$} }
  }
  child {
    node (l) {$\mvconst^{(2)}_{f}$}
      child { node (o) {$\vdots$} }
      child { node (p) {$\vdots$} }
  }
}
;
\path (a) -- (j) node [midway] {\dots};
\path (b) -- (g) node [midway] {\dots};
\path (k) -- (l) node [midway] {\dots};
\path (c) -- (d) node [midway] {\dots};
\path (e) -- (f) node [midway] {\dots};
\path (m) -- (n) node [midway] {\dots};
\path (o) -- (p) node [midway] {\dots};
\end{tikzpicture}
    \caption{An illustration of a refinement tree $T$. Each vertex represents a constituent and each edge indicates a refinement relation.}
    \label{fig:statesp2:reftree}
\end{figure}

The intuitive difference between the fixed quantifier rank case and the arbitrary quantifier rank case is ``resolution": the former considers descriptions of logical possibilities at a fixed resolution whereas the latter considers descriptions of logical possibilities at arbitrary resolution. It is thus natural to think of an infinitely precise description of a logical possibility as the limit of a sequence of increasingly finer descriptions.

Our task in this section is to formalize this intuition of the limit of a convergent sequences of logical descriptions. The contents of this section are rather tedious and not essential to the rest of the paper. The high-level takeaway is that we can associate a tree structure (Figure~\ref{fig:statesp2:reftree}) with the set of constituents $\consts \eqdef \bigcup_{r \in \N} \consts^{(r)}$ where each (infinite) path $(\mvconst^{(r)})_{r \in \N}$ in the tree corresponds to a sequence of constituents ordered by quantifier rank that provide increasingly finer descriptions. The limiting description is the ``endpoint" of the corresponding path and provides an infitely precise description.

\paragraph{Expansion}
We begin by recalling a basic fact about constituents and their expansion to higher quantifier rank.
\begin{proposition}[\cite{hintikka1965distributive}]
Every rank $r$ constituent can be written as a disjunction of rank $r+s$ constituents, \ie, there is an \emph{expansion} relation $\succ^s_r: \consts^{(r+s)} \times \consts^{(r)} \rightarrow \two$ such that any $\mvconst^{(r)}$ can be expressed as
\[
\mvconst^{(r)} \equiv \lOr_{\mvconst^{(r+s)} \succ^s_r \mvconst^{(r)}} \mvconst^{(r+s)} \,.
\]
\end{proposition}

Observe that a constituent can be the expansion of two distinct constituents, in which case it is necessarily unsatisfiable. In symbols, if $\mvconst^{(r+s)} \succ^s_r \mvconst^{(r)}_1$ and $\mvconst^{(r+s)} \succ^s_r \mvconst^{(r)}_2$ then $\nvDash \mvconst^{(r+s)}$.

Additionally, observe that the notion of expansion is consistent with logical implication. More concretely, if $\mvconst^{(r+s)} \succ^s_r \mvconst^{(r)}$ then $\vDash \mvconst^{(r+s)} \rightarrow \mvconst^{(r)}$. When $\mvconst^{(r+s)}$ is satisfiable, we can interpret it as extending the description of a logical possibility denoted by $\mvconst^{(r)}$ to account for $s$ additional individuals. In other words, $\mvconst^{(r+s)}$ is a higher resolution description of $\mvconst^{(r)}$. When $\mvconst^{(r+s)}$ is not satisfiable, we can interpret it as an inconsistent description obtained by extending some description (either consistent or inconsistent).

These two observations highlight a certain asymmetry between satisfiable and unsatisfiable descriptions. Whereas a satisfiable constituent is necessarily in the expansion of another unique satisfiable constituent at a given rank due to mutual exclusivity, an unsatisfiable constituent can be in the expansion of any constituent. We thus have to make a choice about how different logical impossibilities are related to one another in order to prevent two sequences that eventually describe logical impossibilities from merging into each other.

\paragraph{Logical impossibilities}
We use a syntactic criterion to aid with uniquely associating constituents with their expansions.\footnote{The syntactic criterion we use is one of the conditions used in the definition of \emph{trivial inconsistency}~\cite{hintikka1965distributive}. Trivial inconsistency is a syntactic criterion for determining whether or not a constituent is satisfiable that satisfies a completeness property~\cite{hintikka1965distributive}: a constituent that is not satisfiable has a quantifier rank where all of its expansions are trivially inconsistent.} The idea is to define a suitable notion for a constituent $\mvconst^{(r)}$ to be a prefix of another constituent $\mvconst^{(r+1)}$.

The notation
\[
\dagger \mvconst^{(r)} \eqdef \set{ \langle (\exists^E z) \mvaconst^{(r-1)}[z] \rangle^{\mvaconst^{(r-1)}[z]}_{\mvaconst^{(r)}} \ST \mvaconst^{(r-1)}[z] \in \aconsts^{(r-1)}[z]}
\]
gives the \emph{set representation} of a constituent where the notation $\langle \cdot \rangle^{\mvaconst^{(r-1)}[z]}_{\mvaconst^{(r)}}$ indicates that whether or not the formula inside is negated depends on $\mvaconst^{(r-1)}[z]$ and $\mvaconst^{(r)}$. When the quantifier rank is $0$, we have $\dagger \mvconst^{(0)} \eqdef \set{\true}$. The notation $\tilde{\mvconst}^{(r)}$ indicates that we remove the deepest layer of quantification from $\mvconst^{(r)}$, \ie, all formula involving the rank $r$ quantifiers.\footnote{We can see the effect that removing a layer of quantification has on a constituent by rewriting it to only use rank $0$ attributive constituents as
\begin{align*}
    & \mvconst^{(r+1)} \equiv \lAnd_{\mvaconst^{(r-1)}[z_1] \in \aconsts^{(r)}[z_1]} \Big\langle (\exists^E z_1) \mvaconst^{(0)}[z_1] \land \ldots \Big\langle \\
    & \lAnd_{\mvaconst^{(0)}[\bar{z}_{r+1}] \in \aconsts^{(0)}[\bar{z}_{r+1}]} \Big\langle (\exists^E z_{r+1}) \mvaconst^{(0)}[\bar{z}_{r+1}] \Big\rangle^{\mvaconst^{(1)}[\bar{z}_{r}]}_{\mvaconst^{(r+1)}} \ldots \Big\rangle^{\mvaconst^{(r)}[z_1]}_{\mvaconst^{(r+1)}} \,.
\end{align*}
}

We say that $\mvconst^{(r)}$ is a \emph{prefix} of $\mvconst^{(r+1)}$, written $\mvconst^{(r)} \sqsubseteq \mvconst^{(r+1)}$, if $\dagger \mvaconst^{(r)} = \dagger \tilde{\mvaconst}^{(r+1)}$. 
\begin{proposition}
\begin{enumerate}
    \item Every $\mvconst^{(r)}$ is a prefix of some $\mvconst^{(r+1)}$.
    \item If $\mvconst^{(r)} \sqsubseteq \mvconst^{(r+1)}$ and $\cM \vDash \mvconst^{(r+1)}$ for some $\cM$, then $\cM \vDash \mvconst^{(r)}$.
    \item If $\mvconst^{(r)} \sqsubseteq \mvconst^{(r+1)}$ and $\nvDash \mvconst^{(r)}$, then $\nvDash \mvconst^{(r+1)}$.
\end{enumerate}
\label{prop:statesp2:ms:prefix}
\end{proposition}
\begin{proof}
See appendix.
\end{proof}
The first item ensures that every constituent is the prefix of some expansion constituent so that the resolution of its description can increased. The second and third items ensure that prefixes are compatible with logical implication. More concretely, the second ensures that satisfiable expansions have satisfiable prefixes while the third ensures that an unsatisfiable constituent is the prefix of some unsatisfiable expansion.

\begin{algorithm}[t]
\begin{algorithmic}[1]
\Function{edges}{$r$, $\xi$}
\For{$\mvconst^{(r+1)} \in \consts^{(r+1)}$}
    \If{$\mvconst^{(r)} \sqsubseteq \mvconst^{(r+1)}$ for some $\mvconst^{(r)}$}
    \State Add the edge $(\mvconst^{(r)}, \mvconst^{(r+1)})$ to $\xi$.
    \Else
    \State Let $A \eqdef \set{\mvconst^{(r)} \ST \mvconst^{(r+1)} \succ^1_r \mvconst^{(r)}}$.
    \If{$|A| = 0$}
      \State Choose an arbitrary $\mvconst^{(r)} \in \consts^{(r)}$.
      \State Add the edge $(\mvconst^{(r)}, \mvconst^{(r+1)})$ to $\xi^{(r+1)}$. 
    \Else
      \State Choose an arbitrary $\mvconst^{(r)} \in A$.
      \State Add the edge $(\mvconst^{(r)}, \mvconst^{(r+1)})$ to $\xi^{(r+1)}$. 
    \EndIf
    \EndIf
\EndFor
\State \textbf{return} $\xi$
\EndFunction
\end{algorithmic}  
\caption{Algorithm for adding edges to a refinement tree.}
\label{alg:statesp2:ms:edges}
\end{algorithm}
\setlength{\textfloatsep}{0.1cm}

\paragraph{Refinement tree}
Define a sequence of graphs as follows. Let $T_0 \eqdef (\consts^{(0)}, \set{})$ so that the vertices are the constituents of rank $0$ and no edges. Let $T_{r+1} \eqdef T_r \cup (\consts^{(r+1)}, \xi^{(r+1)})$ with edges $\xi^{(r+1)} \eqdef \operatorname{edges}(r, \xi^{(r)})$ where $\operatorname{edges}$ is defined in Algorithm~\ref{alg:statesp2:ms:edges}.
\begin{definition}
The graph $T \eqdef \bigcup_{r \in \N} T_r$ is a \emph{refinement tree}.
\end{definition}
\noindent  We say that $\mvconst^{(r+s)}$ \emph{refines} $\mvconst^{(r)}$ if there is a path from $\mvconst^{(r)}$ to $\mvconst^{(r+s)}$ in $T$.
We check that $T$ is a tree with the desired properties.
\begin{proposition}
For any $r$, $T_r$ is a tree such that $\vDash \mvconst^{(s+1)} \rightarrow \mvconst^{(s)}$ for any $(\mvconst^{(s)}, \mvconst^{(s+1)}) \in \xi^{(r)}$.
\label{prop:statesp2:ms:tree}
\end{proposition}
\begin{proof}
See appendix.
\end{proof}
\begin{proposition}
$T$ is a tree such that $\vDash \mvconst^{(r+1)} \rightarrow \mvconst^{(r)}$ for any $(\mvconst^{(r)}, \mvconst^{(r+1)}) \in E$.
\label{prop:statesp2:ms:prefix2}
\end{proposition}
\begin{proof}
See appendix.
\end{proof}
\noindent The tree structure lends itself to constructing a topology that expresses what an initial segment of refining constituents could converge to.

\paragraph{Topology}
Recall that we can define a topology on a tree by defining a basis of open sets identified by finite initial segments of the tree.\footnote{For background on topology, we refer the reader to~\cite{munkres2000topology}.} Let $\Psi$ be the set of paths of a refinement tree. We write $\bar{\mvconst}^{(r)}$ to indicate the path $(\mvconst^{(0)}, \dots, \mvconst^{(r)})$ in a refinement tree. Then a basic open set has the form
\[
B_{\mvconst^{(r)}} \eqdef \set{\rho \in \Psi \ST \bar{\mvconst}^{(r)} \sqsubseteq \rho}
\]
where $\sqsubseteq$ indicates a prefix relation on refinement tree paths. We can think of a basic open $B_{\mvconst^{(r)}}$ as the description of $\mvconst^{(r)}$ along with all possible ways in which that description can be extended.
\begin{definition}
The \emph{constituent refinement} space is the topological space $(\Psi, \cO(\Psi))$ where the topology $\cO(\Psi)$ is generated by the basis of clopen sets
\[
\set{B_{\mvconst^{(r)}} \ST \mvconst^{(r)} \in \consts^{(r)}} \,.
\]
\end{definition}

\subsection{Hilbert Space}
\label{subsec:statesp2:hilbert}

At the end of the previous section, we obtained a topological space on appropriate sequences of constituents. In this section, we will use the notion of convergence given by the topological space to define a Hilbert space that encodes sentences of arbitrary quantifier rank. Towards this end, we reintroduce metric notions to the topological space so that we can recover a geometry.

The inner product $\inner{\cdot}{\cdot}$ is the fundamental metric notion in a Hilbert space. It intuitively gives a measure of similarity between two elements of a Hilbert space. In our case, we have a space of paths $\Psi$ through a refinement tree. Recall that a path in a refinement tree describes a logical description in finer and finer detail. Thus, if we were to measure the similarity between two such paths, we would intuitively weigh the difference in coarser descriptions (the big picture) more than differences in finer descriptions (the minutiae). This leads us to construct a Hilbert space using $L^2(X, \mu)$, \ie, the space of square summable sequences of $X$ weighted by a measure $\mu$.

\begin{definition}
Let $\hilbsp \eqdef L^2(\Psi, \mvmeas)$ be the \emph{Hilbert space} for sentences of arbitrary quantifier rank where $\mvmeas$ is any measure on $\Psi$'s Borel $\sigma$-algebra. Unless stated otherwise, we assume that $\mvmeas$ is the uniform measure, \ie, $\mvmeas(B_{\mvconst^{(r)}}) = \frac{1}{|\consts^{(r)}|}$.\footnote{Recall that we can uniquely extend a valuation on a clopen basis to a measure on $\sigma$-algebra. For background on measure-theoretic probability, we refer the reader to~\citep{kallenberg2006foundations}.}
\end{definition}
\noindent Recall that $L^2(X, \mu)$ is a Hilbert space so $\hilbsp$ is a Hilbert space.

The Hilbert space spans all first-order sentences. Let $\chi_A(\cdot)$ be the characteristic function over the set $A$.
\begin{proposition}
Every first-order sentence $\mvform^{(r)}$ can be written as
\[
\bm{\mvform}^{(r)} = \sum_{i=1}^n a_i \chi_{B_{\mvconst^{(r)}_i}}(\cdot)
\]
where $\dnf(\mvform^{(r)}) = \set{\mvconst^{(r)}_1, \dots, \mvconst^{(r)}_n}$ and $a_1, \dots, a_n$ are non-zero. Unless stated otherwise, we assume $a_1 = \dots = a_n = |\dnf(\mvform^{(r)})|$.
\end{proposition}
\begin{proof}
Every first-order sentence of rank $r$ can be written as a disjunction of rank $r$ constituents (Proposition~\ref{prop:statesp:vecsp:constituent}).
\end{proof}
\noindent Thus we can think of a first-order sentence $\mvform^{(r)}$ as the \emph{signal} $\bm{\mvform}^{(r)}$ in $\hilbsp$.

\paragraph{Orthogonality}
The inner product $\inner{\cdot}{\cdot}$ is given by
\[
\inner{f}{g} = \int f \cdot g \, d\mvmeas
\]
for $f, g \in \hilbsp$ (\ie, $f$ and $g$ are square integrable).

Although there are an uncountable number of infinitely precise descriptions, there is a sense in which we really only have a denumerable handle on them. 
\begin{proposition}
$\hilbsp$ has a countable basis.
\end{proposition}
\begin{proof}
Let $\cB_r \eqdef \set{\chi_{B_\mvconst^{(r)}} \ST \mvconst^{(r)} \in \consts^{(r)}}$ and $\cB \eqdef \bigcup_{r \in \N} \cB_r$. Then $\cB$ is countable and spans the space.
\end{proof}
\noindent We can apply Gram-Schmidt to $\cB$ to obtain a orthonormal basis.

This concludes the construction of a space for first-order logic. For simplicity, we will largely restrict attention to the case of fixed quantifier rank throughout the rest of the paper as it suffices to illustrate the main ideas of what a ``signals in space" perspective might offer.

\section{Plausibility}
\label{sec:plaus}

\emph{How does deduction affect the ``plausibility" of a logical system?}
In this section, we conceptualize a first-order space as carrying the ``plausibility" of first-order sentences (Section~\ref{subsec:plaus:plaus}). Like its physical counterpart which carries a conserved amount of energy in various forms, a logical system carries a conserved amount of ``plausibility" in various places. We will then see that deduction is an entropic operation (Section~\ref{subsec:plaus:mp}). This provides an alternative explanation for the failure of \emph{logical omniscience}, \ie, why a rational agent fails to know all the consequences of a set of axioms. We restrict attention to sentences of quantifier rank less than $R$ in this section.

\subsection{Plausibility as Energy}
\label{subsec:plaus:plaus}

Energy is a concept that is fundamental to the physical world. In particular, the \emph{law} of conservation of energy states that energy cannot be created or destroyed---it can only be converted from one form into another by some \emph{physical process}. We can thus reason about the effects of physical processes by accounting for the system's energy.

We propose that ``plausibility" is the analogous concept of energy for a logical system. A \emph{logical system} is comprised of a single sentence.\footnote{A logical system consisting of $n$ sentences $\mvform_1, \dots, \mvform_n$ is encoded as the one sentence system $\mvform_1 \land \dots \land \mvform_n$.} We define the \emph{plausibility} of a logical system $\bm{\mvform}$ to be
\[
\cP(\bm{\mvform}) \eqdef \sum_{\mvconst^{(R)} \in \consts^{(R)}} p(\mvconst^{(R)}) \inner{\hat{\mvconst}^{(R)}}{\bm{\mvform}}
\]
where $p$ is some probability mass function on $\consts^{(R)}$.

We thus have the following analogy with energy:
\begin{align*}
    \text{plausibility} & \sim \text{energy} \\
    \text{logical possibilities} & \sim \text{forms of energy} \,.
\end{align*}
Each logical possibility corresponds to a different form of energy (\eg, kinetic). 

The \emph{law of conservation of plausibility} can then be stated as follows: ``Plausibility cannot be created or destroyed---it can only be shifted from one logical possibility to another by some \emph{reasoning process}". As with conservation of energy, conservation of plausibility only holds when we consider the entire system, \ie, $\bm{\true}$. We use the phrase ``reasoning process" informally to suggest that the process is not necessarily deductive.\footnote{For example, \citet{huang2019oltp} shows how a local Bayesian update process can be used to update the assignment of plausibility to logical possibilities.} Nevertheless, deduction is an essential reasoning process and so we turn our attention towards this kind of reasoning next.

\subsection{Deduction}
\label{subsec:plaus:mp}

From the perspective of language, deduction or \emph{modus ponens} (mp) is formulated as the inference rule $\mvform_1 \rightarrow \mvform_2, \mvform_1 \implies \mvform_2$ where $\implies$ transforms the left hand side to the right hand side. From the perspective of space, mp will be formulated as an \emph{operator}.

\paragraph{Modus ponens as an operator}
We use the observation that $\mvform_1 \rightarrow \mvform_2$ is syntactic sugar for $\lnot \mvform_1 \lor \mvform_2$ to express mp as an operator.
\begin{definition}
Define the operator $\textbf{MP}_{\mvform_1 \rightarrow \mvform_2}: \vecsp^{(R)} \rightarrow \vecsp^{(R)}$ as $\textbf{MP}_{\mvform_1 \rightarrow \mvform_2}(\bm{\mvform}) \eqdef (\bm{\mvform} \ominus \bm{\mvform_1}^\perp) \oplus \bm{\mvform_2}$ where:

\begin{tabular}{|l|rl|}\hline
$\land$  & $\mvform_1 \oplus \mvform_2 \eqdef$ & $\sum_{\mvconst^{(R)} \in \consts^{(R)}} \max(\inner{\hat{\mvconst}^{(R)}}{\bm{\mvform_1}}, \inner{\hat{\mvconst}^{(R)}}{\bm{\mvform_2}}) \hat{\mvconst}^{(R)}$ \\
$\lor$   & $\mvform_1 \ominus \mvform_2 \eqdef$ & $\sum_{\mvconst^{(R)} \in \consts^{(R)}} \min(\inner{\hat{\mvconst}^{(R)}}{\bm{\mvform_1}}, \inner{\hat{\mvconst}^{(R)}}{\bm{\mvform_2}}) \hat{\mvconst}^{(R)}$ \\
$\lnot$  & $\mvform^\perp \eqdef$ & $\sum_{\mvconst^{(R)} \notin \dnf(\mvform)} \inner{\hat{\mvconst}^{(R)}}{\bm{\mvform}} \hat{\mvconst}^{(R)}$ \\ \hline
\end{tabular}
\end{definition}
\noindent The additional vector $\bm{\mvform} \ominus \bm{\mvform_1}^\perp$ is $\vec{0}$ when $\bm{\mvform} = \bm{\mvform_1}$. Thus this definition of mp from the space perspective additionally keeps track of how well the antecedent to an implication matches with the current state. It is an over approximation of logical mp.

Analogous to how we can model the evolution of a physical system by considering the sequencing of operators that correspond to physical processes, we can model the evolution of a logical system by considering the effects of sequencing mp operators that correspond to deductive processes.

Observe that $\textbf{MP}_{\mvform_1 \rightarrow \mvform_2}$ is a continuous operator. We can interpret this as another instantiation of the locality of computation: Any sentence that can be deduced is a sequence of local and small steps. As it turns out, $\textbf{MP}_{\mvform_1 \rightarrow \mvform_2}$ is ``entropic" in a sense that we describe next.

\paragraph{Modus ponens is entropic}
We begin with the observation that mp increases plausibility of a logical system.
\begin{proposition}
$\cP(\bm{\mvform_1 \land \hdots \land \mvform_n}) \leq \cP(\textbf{MP}_{\mvform_{n-1} \rightarrow \mvform_n} \circ \dots \circ \textbf{MP}_{\mvform_{1} \rightarrow \mvform_2}(\bm{\mvform_1}))$.
\end{proposition}
\begin{proof}
We get that $\textbf{MP}_{\mvform_{n-1} \rightarrow \mvform_n}(\mvform_{n-1}) \circ \dots \circ \textbf{MP}_{\mvform_{1} \rightarrow \mvform_2}(\mvform_1) = \mvform_n$ by induction on $n$. It then follows that $\cP(\mvform_1 \land \dots \land \mvform_n) \leq \cP(\mvform_n)$ because $\dnf(\mvform_1 \land \dots \land \mvform_n) \subseteq \dnf(\mvform_n)$.
\end{proof}
\noindent Consequently, knowledge of a proof contains less plausibility than knowledge of a theorem statement.

Intuitively, the plausibility of a logical system increases because the number of possibilities we consider possible after an application of mp increases. Note that this fact does not contradict the law of conservation of plausibility. In particular, an application of mp changes the logical system under consideration which may no longer be the total system $\bm{\true}$.

We can rephrase the increase in plausibility as an increase in unnormalized entropy. Recall the unnormalized entropy of a vector $(a_1, \dots, a_n)$ is defined as $U(a_1, \dots, a_n) \eqdef \sum_{i=1}^n a_i (1 - \log(a_i))$. Let $\bm{\mvform}^\ddagger$ be the vector where $\inner{\hat{\mvconst}^{(R)}}{\bm{\mvform}} = 1$ whenever $\mvconst^{(R)} \in \dnf(\mvform)$ and $\inner{\hat{\mvconst}^{(R)}}{\bm{\mvform}} = \alpha$ otherwise where $\alpha > 0$.
\begin{proposition}
\[
U((\bm{\mvform_1 \land \mvform_2})^\ddagger) \leq U((\textbf{MP}_{\mvform_1 \rightarrow \mvform_2}(\bm{\mvform_1}))^\ddagger)
\]
\end{proposition}
\begin{proof}
Observe that $\dnf(\mvform_1) \subseteq \dnf((\mvform \land \lnot \mvform_1) \lor \mvform_2)$. The result thus follows by routine calculation.
\end{proof}
\noindent Because deduction increases (unnormalized) entropy, a deductive process may spread out information over too many possibilities as to render it impossible to find the proof of a conjecture.

\section{Theories}
\label{sec:theory}

\emph{Where is a first-order theory located?}
As we might expect, a first-order theory $T$ forms a subspace. We can visualize the subset of sentences $T^{(r)}$ that have quantifier rank $r$ as a hyperplane sitting in $\vecsp^{(r)}$ (Figure~\ref{fig:theory:subsp}).
\begin{proposition}
Let $T$ be some theory. Then $T^{(r)}$ forms a subspace of $\vecsp^{(r)}$ for every $r$.
\end{proposition}
\begin{proof}
Observe that the subset of satisfiable constituents that satisfy the theory $T$ forms a basis for the subspace.
\end{proof}

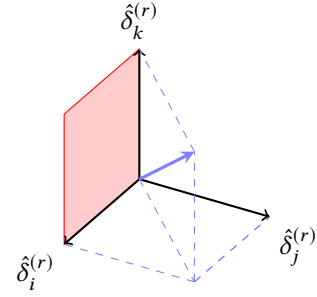
\begin{figure}[t]
    \centering
    \tdplotsetmaincoords{60}{120} 
\begin{tikzpicture} [scale=2, tdplot_main_coords, axis/.style={->,thick}, 
vector/.style={-stealth,blue!50,very thick}, 
vector guide/.style={dashed,blue!50}]

\coordinate (O) at (0,0,0);


\pgfmathsetmacro{\ax}{1}
\pgfmathsetmacro{\ay}{1}
\pgfmathsetmacro{\az}{1}

\coordinate (P) at (\ax,\ay,\az);

\filldraw[
        draw=red,%
        fill=red!20,%
    ]          (0,0,0)
            -- (1,0,0)
            -- (1,0,1)
            -- (0,0,1)
            -- cycle;

\draw[axis] (0,0,0) -- (1,0,0) node[anchor=north east]{$\hat{\mvconst}^{(r)}_i$};
\draw[axis] (0,0,0) -- (0,1,0) node[anchor=north west]{$\hat{\mvconst}^{(r)}_j$};
\draw[axis] (0,0,0) -- (0,0,1) node[anchor=south]{$\hat{\mvconst}^{(r)}_k$};

\draw[vector] (O) -- (P);

\draw[vector guide]         (O) -- (\ax,\ay,0);
\draw[vector guide] (\ax,\ay,0) -- (P);
\draw[vector guide]         (P) -- (0,0,\az);
\draw[vector guide] (\ax,\ay,0) -- (0,\ay,0);
\draw[vector guide] (\ax,\ay,0) -- (0,\ay,0);
\draw[vector guide] (\ax,\ay,0) -- (\ax,0,0);

\end{tikzpicture}
    \caption{A theory $T$ (shaded in red) that is spanned by $\set{\hat{\mvconst}^{(r)}_i, \hat{\mvconst}^{(r)}_j}$. Thus the sentence denoting the vector (blue) is not part of the theory when $\hat{\mvconst}^{(r)}_j$ is not satisfiable.}
    \label{fig:theory:subsp}
\end{figure}

\begin{definition}
The \emph{dimension} of a theory is a function $D: \N \rightarrow \N$ such that $D(r)$ gives the number of satisfiable constituents at rank $r$.
\end{definition}
\noindent The dimension of a theory has can be interpreted as a measure of its complexity. We introduce a complete theory and an undecidable theory next to give examples of the lower and upper limits of a theory's complexity.

\begin{example}[Unbounded dense linear order]
The theory of unbounded dense linear orders (UDLO) expresses the ordering of the real line using the binary relation $<$.\footnote{For reference, the axioms for UDLO are: $(\forall x) (\forall y) x < y \rightarrow \lnot (x = y \lor y < x)$ (antisymmetry), $(\forall x) (\forall y) (\forall z) x < y \land y < z \rightarrow x < z$ (transitivity), $(\forall x) (\forall y) x < y \lor y < x \lor x = y$ (trichotomy), $(\forall x) (\forall y) x < y \rightarrow (\exists z) x < z \land y < z$ (dense), and $(\forall x) (\exists y) (\exists z) y < x \land x < z$ (unbounded).}
\end{example}

UDLO is a complete theory so it is intuitively ``simple".
\begin{proposition}
A complete theory has dimension $D(r) = 1$ the constant one function.
\end{proposition}
\begin{proof}
Let $\mvconst^{(r)}$ be some satisfiable constituent (which exists because the theory is consistent). Then $\lnot \mvconst^{(r)} \equiv \lOr_{\mvconst \in \consts^{(r)} \backslash \set{\mvconst^{(r)}}} \mvconst$ is not satisfiable by completeness of the theory. Thus every other constituent at rank $r$ is not satisfiable.
\end{proof}

To give a flavor of what a constituent looks like, the only satisfiable rank $1$ and $2$ constituents for UDLO are
\begin{enumerate}
    \item $((\exists^E x_1) \, x_1 \nless x_1) \land (\forall^E x_1) \, x_1 \nless x_1$; and
    \item $((\exists^E x_1) \, x_1 \nless x_1 \land E[x_1] \land U[x_1]) \land (\forall^E x_1) \, x_1 \nless x_1 \land E[x_1] \land U[x_1]$
    where $E[x_1] \eqdef (\exists^E x_2) \, x_2 < x_1 \land (\exists^E x_2) \, x_1 < x_2 \land (\exists^E x_2) \, x_2 \nless x_2$ and $U[x_1] \eqdef (\forall^E x_2) \, x_2 < x_1 \lor x_1 < x_2 \lor x_2 \nless x_2$
\end{enumerate}
where $x \nless y$ is syntactic sugar for $\lnot (x < y)$. In words, the only satisfiable rank $2$ constituent describes that either $x_1 < x_2$ or $x_2 < x_1$, \ie, the only two possible linear orderings involving two elements.

\begin{example}[Groups]
The theory of groups expresses groups from abstract algebra and involves one binary operator $\cdot$ expressing group multiplication.\footnote{For reference, the axioms of groups are: $(\exists e) (\forall x) (x \cdot e = x \land e \cdot x = x) \land ((\exists y) x \cdot y = e \land y \cdot x = e)$ (identity) and $(\forall x) (\forall y) (\forall z) (x \cdot y) \cdot z = x \cdot (y \cdot z)$ (transitivity).
Recall that we do not consider function symbols in this paper. Thus the axioms should replace the binary operator $\cdot$ with a ternary predicate $M$ and additional formula asserting that $M$ is a function.}
\end{example}

Groups is not decidable so it is intuitively ``complex".
\begin{proposition}
An undecidable theory has uncomputable dimension.
\end{proposition}
\begin{proof}
A decision procedure for validity of constituents would decide validity of first-order sentences.
\end{proof}

Although the dimension is not computable, we do have a (trivial) upper-bound for the dimension of any theory given by the number of constituents.
\begin{proposition}
The dimension $U(r) = |\consts^{(r)}|$ is an upper bound on the dimension of $D$ of any theory.
\end{proposition}
\begin{proof}
Straightforward.
\end{proof}
\noindent Thus there is a super-exponential upper bound on the dimension of a theory as a function of quantifier rank $r$.

Compared to UDLO which only has one satisfiable constituent at any depth, there are many satisfiable constituents at any depth for groups. Intuitively, these constituents correspond to classifying the kinds of groups that are expressible in a first-order language. The most granular kind of group simply satisfies the axioms, a finer kind of group differentiates abelian groups from non-abelian groups, and so on.

\section{Proofs}
\label{sec:proofs}

\emph{What does orthogonal decomposition reveal about proofs?}
We revisit the notion of proof from the space perspective in this section and explore what one based on orthogonal decomposition would look like (Section~\ref{subsec:proofs:comp}). We will see that each component can be lower bounded by a counter-example (\ie, theorem proving) or upper bounded by an \emph{example} (\ie, model checking) (Section~\ref{subsec:proofs:approx}). This provides an explanation of the utility of examples to the process of proving. Finally, we highlight an application to assigning probabilities to first-order sentences that reflects logical uncertainty and provide a discussion of conjecturing (Section~\ref{subsec:proofs:beliefs}).

\subsection{Component Proofs}
\label{subsec:proofs:comp}

The idea of decomposing a signal into its orthogonal components is one that has numerous applications in engineering and the sciences. It is a powerful method because it enables one to analyze each component separately and combine the results to obtain the behavior of the original signal. For example, we can understand the motion of an object in a 2D plane by separately analyzing its motion along the x-axis and y-axis and superimpose the results.

Decomposition is similar to the usage of \emph{compositionality} in language where we give the meaning of a sentence be composing the meanings of its phrases. The difference between the two is that the former implicitly assumes the pieces to be composed are semantically independent whereas the latter does not. Thus we will be looking at proofs under an additional constraint in this section.

In the case of proving, the assumption of semantic independence means that each piece of the proof corresponds to a distinct logical possibility, and consequently, can be tested for satisfiability independently of one another. We formalize this idea now.
\begin{definition}
A (rank $r$) \emph{component proof system} is a tuple $(\cB, \tau, \chi)$ where
\begin{itemize}
    \item $\cB$ is a choice of orthogonal basis for $\vecsp^{(r)}$;
    \item $\tau: \cB \rightarrow \two$ is a \emph{valuation} such that $\tau(\psi) = 0$ means that $\psi$ is unsatisfiable and $\tau(\psi) = 1$ means that $\psi$ is satisfiable; and
    \item $\chi: \cL \times \cB \rightarrow \two$ is a \emph{component classifier} that indicates whether a sentence contains the corresponding basis element as a component or not. In symbols, $\chi(\mvform^{(r)}, \psi) = 1$ if $\vDash \psi \rightarrow \mvform^{(r)}$ and $\chi(\mvform^{(r)}, \psi) = 0$ otherwise. The component function is given by
    \[
    \chi(\mvform^{(r)}, \mvconst) = \mvconst \in \dnf(\mvform^{(r)})
    \]
    when $\cB$ is the standard basis.
\end{itemize}
\end{definition}
\begin{definition}
A \emph{component proof} of $\mvform$ is a sentence
\[
\lAnd_{\psi \in \cB} \tau(\psi) \land \chi(\mvform^{(r)}, \psi)
\]
where $(\cB, \tau, \chi)$ is a component proof system. Each $\tau(\psi) \land \chi(\mvform^{(r)}, \psi)$ is a \emph{component}.
\end{definition}

\emph{A priori}, neither the valuation $\tau$ nor the component classifier $\chi$ are required to reflect logical validity. We say that a valuation $\tau$ is \emph{sound} if $\tau(\psi) \neq 1$ whenever $\psi$ is not satisfiable and a valuation $\tau$ is \emph{complete} if $\tau(\psi) = 1$ whenever $\psi$ is satisfiable. Similarly, we say that a component classifier is \emph{sound} if $\chi(\mvform^{(r)}, \psi) \neq 1$ whenever $\nvDash \psi \rightarrow \mvform$ and a component classifier is \emph{complete} if $\chi(\mvform^{(r)}, \psi) = 1$ whenever $\vDash \psi \rightarrow \mvform$.

Because neither the valuation nor component classifier are required to reflect logical validity, a proof constructed in a component proof system can exhibit error. In particular, a component proof can be incorrect on some components and correct on others, leading to a ``partially correct" proof. 
\begin{definition}
\begin{itemize}
    \item The \emph{error} of the valuation $\hat{\tau}$ with respect to the ground truth $\tau$ is 
    \[
    E_V \eqdef \sum_{\psi \in \cB} |\hat{\tau}(\psi) - \tau(\psi)| \,.
    \]
    \item The \emph{error} of the component classifier $\hat{\chi}$ with respect to the ground truth $\chi$ on the sentence $\mvform^{(r)}$ is
    \[
    E_{\mvform^{(r)}} \eqdef \sum_{\psi \in \cB} |\hat{\chi}(\mvform^{(r)}, \psi) - \chi(\mvform^{(r)}, \psi)| \,.
    \]
\end{itemize}
The \emph{total error} of a component proof for $\mvform^{(r)}$ is $E_V + E_{\mvform^{(r)}}$.
\end{definition}
\noindent We say that a component proof is \emph{partially correct} if it has non-zero total error.

The error is related to the soundness and completeness of the valuation and component classifiers in the obvious way.
\begin{proposition}
The total error $E_V + E_{\mvform^{(r)}} = 0$ for every sentence $\mvform^{(r)}$ iff both the valuation and component classifier are sound and complete.
\end{proposition}
\begin{proof}
Straightforward.
\end{proof}

The notion of partial correctness is an essential difference between a standard proof and a component proof. As a reminder, a standard proof is a sequence of sentences where \emph{every} sentence follows from the previous one by deduction. One invalid proof step thus invalidates the entire proof---a standard proof is brittle. In contrast, a component proof is somewhat robust because each component is semantically independent of the other. The robustness to error leads us to consider how approximation can be used to construct a component proof.

\subsection{Approximation}
\label{subsec:proofs:approx}

\begin{figure}[t]
    \centering
    \includegraphics{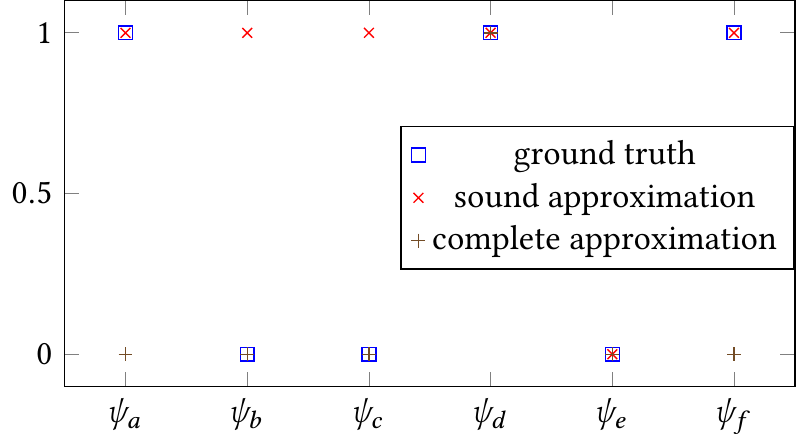}
    \caption{An illustration of approximating a component in a component proof. The idea is to determine each basis element's ground truth value by either upper or lower bounding its value. In this example, we know the ground truth value for the basis elements $\psi_d$ and $\psi_e$ because the upper and lower bounds coincide. For basis elements $\psi_a$, $\psi_b$, $\psi_c$, and $\psi_f$ we do not yet know their ground truth value.}
    \label{fig:proofs:sandwich:sandwich}
\end{figure}

The inspiration for an approximation procedure for constructing a component proof comes from sandwiching a function in analysis: use a simpler family of functions to simultaneously lower bound and upper bound the unwieldy function of interest. In our setting, we can use the family of characteristic functions $\set{\chi_{\mvconst^{(r)}}(\cdot) \ST \mvconst^{(r)} \in \consts^{(r)}}$.

Figure~\ref{fig:proofs:sandwich:sandwich} illustrates the idea behind the approximation procedure for the components of a component proof. Let $F_{\mvform}(\psi) \eqdef \tau(\psi) \land \chi(\mvform, \psi)$.
\begin{description}
    \item[Lower bound] We lower bound $F_{\mvform}(\psi)$ by attempting to show that $\lnot \psi$ is logically valid. We can use any first-order proof system (\eg, an automated theorem prover) to show that $\lnot \psi$ is derivable. This corresponds to showing that this specific component has a counter-example. 
    \item[Upper bound] We upper bound $F_{\mvform}(\psi)$ by (1) ``choosing" a model $\cM$ and then (2) attempting to show  that $\cM \vDash \psi$ (\eg, using a model checker). This corresponds to showing that this specific component has an example.
\end{description}

The approximation procedure reveals an asymmetry in the efficacy of lower bounding (\ie, proving) versus upper bounding (\ie, model checking). In particular, any disjunction of unsatisfiable basis elements will be unsatisfiable as well. Consequently, by refuting $\lOr_{\psi \in B} \psi$ for some subset $B$ of unsatisfiable basis elements, we can eliminate $|B|$ basis elements in one go.

In contrast, the model checking approach can only show that one component is satisfiable at a time due to the mutual exclusivity of the basis. Moreover, it requires the user to devise a concrete example. Consequently, from the perspective of a component proof, it should be more effective to put more effort into proving as opposed to model checking.

\subsection{Logical Uncertainty}
\label{subsec:proofs:beliefs}

\renewcommand{\kbldelim}{(}
\renewcommand{\kbrdelim}{)}

Intuitively, we would expect proving and assigning \emph{logical uncertainty}---assessing how probable a conjecture is---to be related tasks. For instance, we might believe that a conjecture is more likely to be true than false even when we have no clear path for proving it.

In this section, we show how an approximation procedure for constructing component proofs gives rise to a probability distribution on first-order sentences such that logically equivalent sentences are not necessarily assigned the same probability. This provides another take on the failure of logical omniscience that complements the one given in Section~\ref{sec:plaus} where we saw that deduction was an entropic operation. We begin by introducing a representation of a theorem proving agent's knowledge.

\paragraph{Representation}
Let $\operatorname{refute}: \cL \rightarrow \two$ be a function that soundly approximates the deduction relation, \ie, $\operatorname{refute}(\mvform) = \true$ implies that $\nvDash \mvform$. We can think of $\operatorname{refute}(\mvform)$ as a resource-bound prover. Similarly, let $\operatorname{check}_\cM \cL \rightarrow \two$ be a function that soundly approximates a model checker, \ie, $\operatorname{check}_\cM(\mvform) = \true$ implies that $\cM \vDash \mvform$. As before, we can think of $\operatorname{check}_\cM(\mvform)$ as a resource-bound model checker.

Suppose we have sentences $\mvform^F_1, \dots, \mvform^F_N$ that are known to be false, knowledge of models $\cM_1, \dots, \cM_M$, and sentences $\mvform_1, \dots, \mvform_L$ whose satisfiability is not known. We can interpret $\lnot \mvform^F_1, \dots, \lnot \mvform^F_N$ as theorems, $\cM_1, \dots, \cM_M$ as common examples, $\mvform_1, \dots, \mvform_L$ as conjectures that are specific to the agent. (We will see why we include conjectures momentarily.) Let $\mvform_0 \eqdef \lAnd_{i = 1}^N \mvform^F_i$. Let $N_{ij} \eqdef 1$ when $\operatorname{refute}(\mvform_i \rightarrow \mvform_j) = \true$ and
\[
N_{ij} \eqdef \frac{|\set{\cM_k \ST 0 \leq k \leq M, \operatorname{check}_{\cM_k}(\mvform_i \rightarrow \mvform_j) = \true}|}{M + \alpha}
\]
otherwise where $0 \leq i \leq L$, $0 \leq j \leq L$, and $\alpha > 0$.
\begin{definition}
The matrix
\[
  K \eqdef \kbordermatrix{
    & \mvform_0 & \mvform_1 & \mvform_2 & \dots & \mvform_L \\
    \mvform_0 & 1 & 1 & 1 & \dots & 1 \\
    \mvform_1 & N_{10} & 1 & N_{12} & \dots & N_{1L} \\
    \mvform_2 & N_{20} & N_{21} & 1 & \dots & N_{2L} \\
    \vdots & \vdots & \vdots & \vdots & \ddots & \vdots \\
    \mvform_L & N_{L0} & N_{L1} & N_{L2} & \dots & 1 \\
  }
\]
represents the \emph{knowledge} of an agent $\cA$.
\end{definition}

Each row and column of the matrix is associated with the sentences $F, \mvform_1, \dots, \mvform_L$. Each entry of the matrix can be thought of as the probability that the row sentence $\mvform_r$ implies the column sentence $\mvform_c$, \ie, $N_{rc} \approx \Pr(\vDash \mvform_r \rightarrow \mvform_c)$. Thus the first row of the matrix is all $1$ (because false statements imply any statement) and the first column of the matrix gives the probability that the given sentence is false (\ie, $N_{r0} \approx \Pr(\nvDash \mvform_r)$). Note that $N_{ij} \neq N_{jk}$ in general (unless we have an iff), and consequently, we should not think of the matrix $K$ as a covariance matrix. Nevertheless, there is a sense in which we can think of each entry as a measure of logical overlap, \ie, logical correlation.

\paragraph{Conjectures}
Consider the subspace of $\vecsp^{(r)}$ spanned by $B \eqdef \set{F, \mvform_1, \dots, \mvform_L}$ where $r$ is the maximum rank of all the sentences involved. We call $B$ a \emph{conjecturing basis}.

In general, $B$ will not be (1) linearly independent nor (2) contain $\true$. In the setting of first-order logic, the first is not an issue for the purpose of determining validity. The second, however, is problematic because it means that there are certain sentences whose validity an agent will not be able to determine unless it considers alternative conjectures. We can thus view conjecturing from the perspective of linear algebra as finding a useful conjecturing basis.

An application of principal component analysis (pca) to $K$ gives an ordering on the agent's conjectures with respect to $K$'s eigenvalues. Because the matrix $K$ is constructed with respect to the satisfiability of an agent's models, we can interpret the sentences with the lowest eigenvalues as those that have no known examples to the sentences with the highest eigenvalues as those with no known counterexamples. Thus there is an exploration versus exploitation tradeoff in determining which conjectures to invest time in. An agent that wishes to explore should invest in the sentences with the smallest eigenvalues while an agent that wishes to exploit should invest in the sentences with the largest eigenvalues.

\paragraph{Beliefs and probabilities}
We return to the problem of defining probabilities on sentences, starting with a proxy for them.
\begin{definition}
An agent with knowledge $K$ has \emph{belief}
\[
\belief(\mvform_i) \eqdef 1 - N_{i0} \,.
\]
\end{definition}
\noindent Note that an agent's beliefs do not represent probabilities because they are not normalized.

Although an agent's beliefs are not probabilities, they can be converted into them. 
\begin{definition}
Define a valuation on basic opens by on induction on quantifier rank $r$ as $\nu(B_{\mvconst^{(0)}}) \eqdef 1$ and 
\[
\nu(B_{\mvconst^{(r+1)}}) \eqdef \begin{cases}
0 & \mbox{$\nu(B_{\mvconst^{(r)}}) = 0$} \\
\frac{\belief(\mvconst^{(r+1)})}{\sum_{\mvconst^{(r+1)} \in \child(\mvconst^{(r)})} \belief(\mvconst^{(r+1)})} \nu(B_{\mvconst^{(r)}}) & \mbox{otherwise}
\end{cases}
\]
where $\mvconst^{(r)}$ is parent of $\mvconst^{(r+1)}$ and $\child(\mvconst^{(r)})$ denotes $\mvconst^{(r)}$'s children nodes.
\end{definition}
\begin{proposition}
$\nu$ defines unique probability measure on $\Psi$.
\end{proposition}
\begin{proof}
We see that $\nu$ defines a valuation on the basic clopens of $\Psi$ such that $\nu(B_{\mvconst^{(r)}}) = \sum_{\mvconst^{(r+1)} \in \child(\mvconst^{(r+1)})} \nu(B_{\mvconst^{(r+1)}})$ for any $\mvconst^{(r)}$ by induction on $r$. Thus we can extend this to a measure in the usual way.
\end{proof}

Because an agent only has a finite amount of knowledge, it follows that an agent's probabilities are not logically omniscient.
\begin{proposition}[Failure of logical omniscience]
There are agents that are not logically omniscient.
\end{proposition}
\begin{proof}
Consider the agent with matrix $K = (1)$ and knows no models. Then it assigns positive probability to unsatisfiable constituents so it cannot be logically omniscient.
\end{proof}

\section{Size}
\label{sec:size}

\emph{How much space does a first-order theory occupy?}
In this section, we explore two views of a ``models as boundary" principle that highlight the difficulty of approximating a theory. In the first, we view distinct models as spanning the theory and quantify the accuracy of approximation as a function of the number of distinct models (Section~\ref{subsec:size:span}). In the second, we view the valuation for some theory's basis as a Boolean hypercube (Section~\ref{subsec:size:iso}). A familiar (edge) isoperimetric principle then says that a theory with small ``volume" (\ie, `variance) can still have large ``perimeter" (\ie, many models).

\subsection{Approximate Spanning}
\label{subsec:size:span}

The first view we have of ``models as boundary" is that a theory's models span it. Intuitively, knowledge of distinct models suffices to describe the theory because each model describes exactly one basis element. It then follows that the number of distinct models we require as a function of quantifier rank is exactly the dimension of the theory.

How many models do we need to approximately span a theory? As it turns out, it depends on how we represent what we know to be true. Let $\bm{\true_U} \eqdef \sum_{\mvconst^{(r)} \in \consts^{(r)}} \frac{1}{|\consts^{(r)}|}\hat{\mvconst^{(r)}}$ be a representation of truth for an \emph{uninformed} agent. Let $\bm{\true_O} = \sum_{\hat{\mvconst}^{(r)} \in \consts^{(r)}|_+} \frac{1}{D(r)}\hat{\mvconst}^{(r)}$ where $\consts^{(r)}|_+$ are the satisfiable constituents of rank $r$ be a representation of truth for an \emph{omniscient} agent.
\begin{proposition}
\begin{enumerate}
    \item There exist $k$ sentences $\mvform^{(r)}_1, \dots, \mvform^{(r)}_k$ such that
    \[
    \lVert \bm{\true_U} - \frac{1}{k}\sum_{j=1}^k \bm{\mvform}^{(r)}_j \rVert_2 \leq \sqrt{\frac{2}{k|\consts^{(r)}|}} \,.
    \]
    \item There exist $k$ sentences $\mvform^{(r)}_1, \dots, \mvform^{(r)}_k$ such that
    \[
    \lVert \bm{\true_O} - \frac{1}{k}\sum_{j=1}^k \bm{\mvform}^{(r)}_j \rVert_2 \leq \sqrt{\frac{2}{k D(r)}} \,.
    \]
\end{enumerate}
\end{proposition}
\begin{proof}
\begin{enumerate}
    \item Let $\set{\mvconst^{(r)}_1, \dots, \mvconst^{(r)}_N} = \consts^{(r)}$. The set $\set{\sum_{i=1}^N a_i \hat{\mvconst}^{(r)}_i \ST a_1, \dots a_N \in [0, 1/|\consts^{(r)}|]}$ forms a convex subset of $\mathbb{V}^{(r)}$ with diameter $\sqrt{2/|\consts^{(r)}|}$ that contains $\bm{\true_U}$.\footnote{The diameter of a convex set $S$ is $\sup \set{\lVert x - y\rVert_2 \ST x, y \in S}$.} The result follows by an application of the approximate Catheodory theorem.
    \item Similar to item $1$.
\end{enumerate}
\end{proof}

At first blush, it appears that a theory can be ``accurately" spanned by a constant number of models independent of its dimension. Nevertheless, we should keep in mind that logical validity is brittle no matter how it is represented---being off by any $\epsilon > 0$ is enough to lose it. 

Comparing the results for $\bm{\true_U}$ versus $\bm{\true_O}$, we see that it is easier to approximate $\bm{\true_U}$ than it is to approximate $\bm{\true_O}$ (\ie, smaller $k$) because $D(r) \leq |\consts^{(r)}|$. Put another way, it is easier to approximate the beliefs of an uninformed agent compared to the beliefs of an omniscient agent. 

Our ability to approximate is different if we know models of a theory.
\begin{proposition}
Let $\cM_1 \vDash \mvform^{(r)}_1, \dots, \cM_k \vDash \mvform^{(r)}_k$ be $k$ distinct models of a theory that satisfy some $\mvform^{(r)}$. Then 
\begin{align*}
\lVert \bm{\true_U} - \bm{\mvform}^{(r)} \rVert_2 & \leq \sqrt{\frac{|\consts^{(r)}|^2 - 2N|\consts^{(r)}| + N^2}{N|\consts^{(r)}|^2}} \\
\lVert \bm{\true_O} - \bm{\mvform}^{(r)} \rVert_2 & \leq \sqrt{\frac{D(r)(D(r) - 2k) + kN}{ND(r)^2}}
\end{align*}
where $\bm{\mvform}^{(r)} \eqdef \sum_{\mvconst^{(r)} \in \dnf(\mvform^{(r)})} \frac{1}{N}\mvconst^{(r)}$ and $N = |\dnf(\mvconst^{(r)})|$.
\end{proposition}
\begin{proof}
By direct calculation.
\end{proof}
\noindent Let us unpack the assertion. Observe that when an agent has uninformative beliefs, the error depends on how well $\mvform^{(r)}$ approximates the $k$ known models. When an agent has omniscient beliefs, we have that the information content of a theory is approximately encoded by its models, \ie, its examples. Intuitively, this makes sense as the more distinct kinds of models that one knows, the better the approximation to the ground truth.
\begin{example} 
For a complete theory, we have $D(r) = 1$. Thus the above result says that we have maximal error if we know no models and zero error if we know a model. For instance, we can fully describe a complete theory like UDLO with either the models $(\R, <)$ or $(\Q, <)$.
\end{example}
\begin{example}
When $D(r)$ is unknown as in the case of an undecidable theory, then the quality of approximation depends on the theory. If $D(r) \gg k$, then the error term is essentially $1/\sqrt{N}$. Another way to put this is that if one knows only a few examples of ``complex" theory, then $\mvform^{(r)}$ should be chosen so that leaves as many possibilities as open in order to obtain a good approximation.
\end{example}

\subsection{Hypercube Isoperimetric Principle}
\label{subsec:size:iso}

The second view of ``models as boundary" takes inspiration from the Boolean hypercube and Boolean Fourier analysis (\eg, see~\cite{o2014analysis}). The basic idea is that we view the valuation for some theory's basis as a Boolean hypercube with vertices labeled according to whether they are satisfiable or not. We make this concrete now.

\paragraph{Hypercube}
A \emph{labeled Boolean hypercube} of dimension $N$ is a tuple $(V, E, \ell)$ where $V \eqdef \two^N$ is a set of vertices of the length $N$ bit-strings, $E \eqdef \set{ \set{x, y} \ST d(x, y) = 1}$ is a set of edge containing those bit-string with Hamming distance one (\ie, $d(x, y) = |\set{n \ST x_n \neq y_n}|$), and $\ell: V \rightarrow \set{-1, 1}$ is a \emph{labeling} that assigns to each vertex a Boolean value. We write $x^{\lnot i}$ to indicate that we flip the $i$-th bit in the bit-string $x$. A \emph{boundary edge} along the $i$-th dimension is then any $x$ such that $\ell(x) \neq \ell(x^{\lnot i})$. The \emph{influence} at coordinate $i$ of $f: \two^N \rightarrow \set{-1, 1}$ is
\[
\textbf{Inf}_i[f] \eqdef \sum_{x \in V} \frac{1}{2^N} \bm{[}f(x) \neq f(x^{\lnot i})\bm{]}
\]
where the notation $\bm{[}\cdot\bm{]}$ indicates an Iverson bracket gives the fraction of boundary edges along the $i$-th dimension. The \emph{total influence} $\textbf{Inf}[f] \eqdef \sum_{i=1}^N \textbf{Inf}_i[f]$ then gives the total fraction of boundary edges.

Our next task is to associate each bit-string with a constituent. Recall that a constituent has the form
\[
\mvconst^{(r)}_s \eqdef \lAnd_{\mvaconst[z] \in \aconsts^{(r-1)}[z]} (\pm)^{s(\mvaconst[z]) } (\exists^E z) \mvaconst[z]
\]
where $s: \aconsts^{(r-1)}[z] \rightarrow \two$. Thus a constituent can be identified with a bit-string of length $|\aconsts^{(r-1)}[z]|$ corresponding to $s$ which forms a vertex of a Boolean hypercube of dimension $|\aconsts^{(r-1)}[z]|$. A valuation on $\tau: \consts^{(r)} \rightarrow \set{-1, 1}$ where $\tau(\mvconst^{(r)}) = -1$ if $\mvconst^{(r)}$ is satisfiable and $\tau(\mvconst^{(r)}) = 1$ otherwise is then a labeling for the \emph{constituent hypercube}.\footnote{We follow the convention that $-1 \mapsto 1$ and $1 \mapsto 0$ when converting between $\set{-1, 1}$ and $\two$ (\eg, see~\cite{o2014analysis}).}

A boundary edge along the $i$-th dimension on the constituent hypercube corresponds to one where changing whether or not an individual that satisfies the $i$-th constituent of rank $r-1$ exists or not flips the valuation. In other words, we have a boundary edge if we cross between those constituents with models and those without models. As we will see next through a familiar edge isoperimetric principle, the number of crossings can be quite large even if the theory has small variance (\ie, small ``volume").

\paragraph{Edge isoperimetric principle}
We begin by calculating the variance of a theory's valuation $\tau$.
\begin{proposition}
\[
\textbf{Var}[\tau] = 1 - \left( \frac{|\consts^{(r)}| - 2D(r)}{|\consts^{(r)}|} \right)^2 \,.
\]
\end{proposition}
\begin{proof}
We have $\E[\tau] = \frac{D(r)}{|\consts^{(r)}|} - \frac{|\consts^{(r)}| - D(r)}{|\consts^{(r)}|}$ and $\E[\tau^2] = \frac{D(r)}{|\consts^{(r)}|} + \frac{|\consts^{(r)}| - D(r)}{|\consts^{(r)}|} = 1$. The result follows by routine calculation.
\end{proof}

We obtain edge isoperimetric principles by recalling two classic results from Boolean Fourier analysis.
\begin{proposition} We have
\begin{itemize}
    \item $\textbf{Var}[\tau] \leq \textbf{Inf}[\tau]$ (Poincare's inequality~\citep[][pg. 36]{o2014analysis}) and
    \item $ \Omega \! \left( \frac{\log n}{n}\right) \textbf{Var}[\tau] \leq \max_i \textbf{Inf}_i[f]$ (KKL theorem~\citep[][pg. 260]{o2014analysis}).\footnote{As a reminder, $\Omega(\cdot)$ is an asymptotic lower-bound.}
\end{itemize}
\end{proposition}

\begin{example}
For a complete theory, the variance is $\frac{4}{|\consts^{(r)}|}(1 - \frac{1}{\consts^{(r)}})$ which is essentially zero. Observe that every coordinate is influential so that $\textbf{Inf}[\tau] = \max_i \textbf{Inf}_i[\tau] = 1$. 
\end{example}

\section{Related Work}
\label{sec:rel}

The primary building blocks for this paper come from Hintikka's work on distributive normal forms~\cite{hintikka1965distributive,hintikka1970surface,hintikka1973logic} and Huang's work~\cite{huang2019oltp} on using them to assign probabilities to first-order sentences. Due to the nature of our exploration, we have also touched upon a variety of topics in bits and pieces. We highlight related work in some of the areas that we touch upon so we can compare and contrast perspectives.

\paragraph{Logic and geometry}
The connections between logic and geometry have been recognized in the literature, toposes being a prime example (\eg, see~\cite{maclane1992sheaves}). Our approach is more specialized and concrete as we work with first-order logic directly and look at as a linear space.

High-dimensional phenomena such as phase transitions (\eg, see~\cite{gent1994sat}) and isoperimetric principles (\eg, see~\cite{o2014analysis}) have also been considered in the context of propositional logic. We explored some high-dimensional phenomena for first-order logic.

\paragraph{Logical uncertainty}
Logical uncertainty and the problem of logical omniscience is a problem of interest in philosophy (\eg, see~\cite{hacking1967slightly,corfield2003towards,parikh2010sentences,seidenfeld2012uncertainty}) and there have been many solutions proposed for addressing it. One approach proposes assigning probabilities to sentences such that logically equivalent statements are not necessarily assigned the same probability (\eg, see~\cite{garrabrant2016logical} and~\cite{huang2019oltp}).\footnote{There are also approaches focused on assigning probabilities to sentences (\eg, see~\cite{gaifman1964concerning} and~\cite{scott1966assigning} for first-order sentences and ~\cite{hutter2013probabilities} for higher-order sentences).} Another approach syntactically models an agent's knowledge (\eg, see~\cite{gaifman2004reasoning} and~\cite{garber1983old}). A third approach models an agent's reasoning ability as bounded (\eg, see~\cite{icard2014mind} and~\cite{bjerring2018dynamic}). Our approach has aspects of all three. We model the agent's knowledge directly and show how it leads to probabilities that fail logical omniscience. We explain bounded rationality by the entropic nature of deduction.

\paragraph{Conjecturing}
A popular approach for automatic conjecture generation is based on enumerating candidate sentences and pruning the list using heuristics (\eg, the shape of the sentence) and model-checking (\eg, keep sentences that have no known counter-example)~\cite{lenat1976artificial,fajtlowicz1988conjectures,haase1990invention,colton1999automatic}. We suggest using pca with respect to an agent's knowledge, which is constructed with model-checking, as a heuristic. \citet{huang2019oltp} proposes conjecturing as hypothesis selection with respect to a probability distribution on first-order sentences.

\section{Conclusion}
\label{sec:concl}

In summary, we explored some consequences of a shift in perspective to ``signals in space" for first-order logic. We close with some thoughts on what a signal perspective of first-order logic provides to guide further research.

One view of the signal encoding of sentences seen in this paper is that it provides another reduction of first-order logic to a propositional setting. It would be interesting to see how techniques typically applied in the propositional setting apply to the representation of first-order logic introduced here. In particular, we have already seen examples of how Boolean Fourier analysis apply in this paper.

A second view of the signal encoding is that it provides another information theoretic take on first-order logic that counts semantic possibilities as opposed to syntactic aspects of the sentences themselves. It would be interesting to see what other aspects of first-order logic that information theoretic ideas can reveal in addition to the one we saw that deduction is entropic.

Lastly, we saw how a signal perspective enables orthogonal decomposition, a technique that is essential in the sciences and statistical inference more generally, to be applied to a logical system. Thus the ideas in this paper may be useful in illuminating the application of machine learning to theorem proving (\eg, see~\cite{kaliszyk2014machine,neurosat,holstep,premise,guided,mctsrl,huang2018gamepad,huang2019oltp}). In particular, the ``models as boundary" principle highlights the difficulty of approximating a theory and so we have a ``negative" result. On a more encouraging note, we have seen how model-checking and theorem proving can be combined in a principled manner to construct a theorem prover.

\begin{acks}                            
  We would like to thank Dawn Song for support to pursue this direction of research.
\end{acks}

\bibliography{references}

\appendix

\section{Appendix}
\label{sec:appendix}

\begin{proposition}[Proposition~\ref{prop:statesp2:ms:prefix}]
\begin{enumerate}
    \item Every $\mvconst^{(r)}$ is a prefix of some $\mvconst^{(r+1)}$.
    \item If $\mvconst^{(r)} \sqsubseteq \mvconst^{(r+1)}$ and $\cM \vDash \mvconst^{(r+1)}$ for some $\cM$, then $\cM \vDash \mvconst^{(r)}$.
    \item If $\mvconst^{(r)} \sqsubseteq \mvconst^{(r+1)}$ and $\nvDash \mvconst^{(r)}$, then $\nvDash \mvconst^{(r+1)}$.
\end{enumerate}
\end{proposition}
\begin{proof}
\begin{enumerate}
    \item When $r = 0$, then $\mvconst^{(0)}$ is the prefix of every $\mvconst^{(1)}$. Suppose $r \geq 1$. Observe that $|\set{\tilde{\mvconst}^{(r+1)} \ST \mvconst^{(r+1)} \in \consts^{(r+1)}}| = 3^{|\aconsts^{(r-1)}[z_1]|}$ because the third possibility for a deletion is that it contains both $\mvaconst^{(r-1)}[z_1]$ and its negation. Meanwhile $|\consts^{(r)}| = 2^{|\aconsts^{(r-1)}[z_1]|}$. The result follows by the pigeonhole principle.
    \item We have $\cM \vDash \langle\mvaconst^{(r)}[z_1]\rangle^{\mvaconst^{(r)}[z_1]}_{\mvaconst^{(r+1)}}$ for every $\mvaconst^{(r)}[z_1] \in \dagger \mvconst^{(r+1)}$ because $\cM \vDash \mvconst^{(r+1)}$. Moreover, we have $\cM \vDash \langle \mvaconst^{(r-1)}[z_1] \rangle^{\mvaconst^{(r-1)}[z_1]}_{\mvconst^{(r)}}$ for every $\mvaconst^{(r-1)}[z_1] \in \dagger \mvconst^{(r)}$ because $\mvconst^{(r)} \sqsubseteq \mvconst^{(r+1)}$. Thus $\cM \vDash \mvconst^{(r)}$ as desired.
    \item Similar to item $2$.
\end{enumerate}
\end{proof}

\begin{proposition}[Proposition~\ref{prop:statesp2:ms:tree}]
For any $r$, $T_r$ is a tree such that $\vDash \mvconst^{(s+1)} \rightarrow \mvconst^{(s)}$ for any $(\mvconst^{(s)}, \mvconst^{(s+1)}) \in \xi^{(r)}$.
\end{proposition}
\begin{proof}
By induction on $r$. The base case is trivial. In the inductive case, we have that $T_r$ is a tree such that $\vDash \mvconst^{(s)} \rightarrow \mvconst^{(s+1)}$ for any $(\mvconst^{(s)}, \mvconst^{(s+1)}) \in \xi^{(r)}$. In order to show that $T_{r+1}$ is a tree, it suffices to show that every $\mvconst^{(r)} \in \consts^{(r)}$ has an outgoing edge and every $\mvconst^{(r+1)} \in \consts^{(r+1)}$ has a unique incoming edge (because of the inductive hypothesis). This follows from Proposition~\ref{prop:statesp2:ms:prefix} item $1$ and the definition of $\operatorname{edges}$.
    
It remains to show that $\vDash \mvconst^{(s)} \rightarrow \mvconst^{(s+1)}$ for any $(\mvconst^{(s)}, \mvconst^{(s+1)}) \in \xi^{(r+1)}$. By the inductive hypothesis, it suffices to show that $\vDash \mvconst^{(r)} \rightarrow \mvconst^{(r+1)}$ for any $(\mvconst^{(r)}, \mvconst^{(r+1)}) \in \xi^{(r+1)}$. We proceed by case analysis on the edge set following the cases in Algorithm~\ref{alg:statesp2:ms:edges}.
    
Suppose there is an edge $(\mvconst^{(r)}, \mvconst^{(r+1)})$ such that $\mvconst^{(r)} \sqsubseteq \mvconst^{(r+1)}$. If $\mvconst^{(r+1)}$ is not satisfiable then there is nothing to show. If $\mvconst^{(r+1)}$ is satisfiable then $\mvconst^{(r)}$ is satisfiable as well by Proposition~\ref{prop:statesp2:ms:prefix} item $2$. Thus $\vDash \mvconst^{(r+1)} \rightarrow \mvconst^{(r)}$. 
    
Suppose there is an edge $(\mvconst^{(r)}, \mvconst^{(r+1)})$ such that $\mvconst^{(r)} \nsucc^1_r \mvconst^{(r+1)}$. Then $\mvconst^{(r+1)}$ is necessarily unsatisfiable so there is nothing to show.
    
Suppose there is an edge $(\mvconst^{(r)}, \mvconst^{(r+1)})$ such that $\mvconst^{(r)} \succ^1_r \mvconst^{(r+1)}$. Then the result follows as $\succ^s_r$ respects implication.
\end{proof}

\begin{proposition}[Proposition~\ref{prop:statesp2:ms:prefix2}]
$T$ is a tree such that $\vDash \mvconst^{(r+1)} \rightarrow \mvconst^{(r)}$ for any $(\mvconst^{(r)}, \mvconst^{(r+1)}) \in E$.
\end{proposition}
\begin{proof}
Consider the set $\set{T_r \ST r \in \N} \cup \set{T}$ ordered by subgraph inclusion. We show that every chain $\cC$ has an upper bound, namely $\bigcup_{t \in \cC} t$. Let $\mvconst^{(r)}$ be some vertex of $\bigcup_{t \in \cC} t$. Then it is in some $T_r$ so it is identified by a unique path because $T_r$ is a tree by Proposition~\ref{prop:statesp2:ms:tree}. Moreover, let $(\mvconst^{(r)}, \mvconst^{(r+1)})$ be some edge of $\bigcup_{t \in \cC} t$. Again, it is an edge of some $T_r$ so that $\vDash \mvconst^{(r+1)} \rightarrow \mvconst^{(r)}$ as desired. By Zorn's lemma, there is a maximal element which in this case is exactly $T$.
\end{proof}

\end{document}